\documentclass{svjour2}                    
\smartqed  
\usepackage{graphicx}
\usepackage[mathscr]{eucal}
\usepackage{amsfonts}
\usepackage{amsmath,amssymb,amscd}
\usepackage{subfigure}
\usepackage{graphicx,color}

\input{epsf}
\newcommand{\beq}{\begin{equation}}
\newcommand{\eeq}{\end{equation}}

 \journalname{Celestial Mechanics and Dynamical Astronomy}

\begin{document}

\title{Seven-body central configurations}

\subtitle{A family of central configurations in the spatial seven-body problem}


\author{Marshall Hampton\and Manuele Santoprete}


\institute{M. Hampton \at
             Department of Mathematics and Statistics,
University of Minnesota - Duluth,
140 CCtr, 10 University Dr.,
Duluth, MN 55812 USA. \\
              \email{mhampton@d.umn.edu}           
           \and
           M. Santoprete \at
              Department of Mathematics,
 Wilfrid Laurier University, Waterloo, ON, Canada\\
 \email{msantoprete@wlu.ca}
}

\date{Received: date / Accepted: date}
\maketitle

\begin{abstract}
The main result of this paper is the existence of a new family of central configurations in the Newtonian spatial seven-body problem. This family is unusual in that it is a {\it simplex stacked central configuration}, i.e the bodies are arranged as concentric three and two dimensional simplexes. 

\keywords{Central configurations\and relative equilibria\and N-body problem\and celestial mechanics.
}
 \subclass{70F15\and 70F10}
\end{abstract}





\section{Introduction}
The Newtonian $n$-body problem concerns the motion of $n$ particles
with masses $m_i\in{\mathbb R}^+$ and positions $q_i\in{\mathbb
R}^d$, where $i=1,\ldots,n$. The motion is governed by Newton's
law of motion 
\beq
\ddot q_i+\gamma_i=0
\eeq 
where
\beq
\gamma_i=\sum_{k\neq i}m_k R_{ik}(q_i-q_k.) \label{eqgamma}\eeq 
and $R_{ik}=\|q_i-q_k\|^{-3}$. The {\it center of mass} $q_G$ of the system is defined by the relation $\sum m_i(q_i-q_G)=0$.

The $n$ particles form a {\it central configuration} if there exists a $\lambda\in{\mathbb R}$ such that $\gamma_i=\lambda(q_i-q_G)$ (Wintner 1941).

Several aspects of the $n$-body problem motivate the study of central configurations. 
First of all, for suitable initial velocities, central  configurations may define a {\it homothetic} motion (i.e. a motion where  only the size of the configuration can change)
or, in more generality, a {\it homographic} motion (a motion where the configuration can rotate and change size). In the case of a strictly three-dimensional motion homographic motions that are not homothetic are forbidden. In the planar $n$-body problem, instead, there exist homographic motions which are {\it relative equilibria} (i.e. motions where the configuration rotates but does not change size). 

Secondly, for the $n$-body problem, every motion starting or ending in total collision is asymptotic to an homothetic motion and has a central configuration as normalized limit (Saari 1980).

Finally, central configurations appear as a key point when studying the topological changes of the integral manifolds. The integral manifold $I_{hc}$ of the $n-$body problem is the set of points having energy $h$ and angular momentum $c$. Changes in the topology of the integral manifold are caused by central configurations. 
Topological informations on the integral manifold may provide qualitative informations about the motion of the bodies. For further details and references see Albouy (1993).

In this paper we will exhibit a new family of spatial central configurations in the seven body problem. This family is unusual in that  it is a {\it simplex stacked central configuration}, namely a configuration formed by an equilateral triangle contained in a tetrahedron (see figure \ref{tetrahedron}). Examples of trivial {\it stacked central configurations} (i.e. configurations where a subset of the points is a central configuration) are well known. For instance the  
configurations where $N$ particles lie at the vertices of a regular $N-$gon  are stacked central configurations when $N$ is not a prime number. Other simple examples are given by the ``rosette central configurations", i.e. planar central configurations where $N$ particles of mass $m_1$ lie at the vertices of a $N-$gon, $N$ particles of masses $m_2$ lie  at the vertices of another $N$-gon rotated an angle of $\pi/n$ from the other, and a particle of mass $m_0$ lies at the center  of the two $N-$gons (see Sekiguchi (2004) and Lei and Santoprete (2006)).

The  first nontrivial example of stacked central configuration, where two bodies on a line are contained in an equilateral triangle,  was found only recently (see Hampton (2005)). 
The new family of central configurations discussed in the present  paper seems to be the natural generalization of the five body one found in Hampton (2005).

The results of this paper and those of Hampton (2005) suggest that such families of stacked central configurations might generalize to higher dimensions, with $n + m + 2$ bodies in ${\mathbb R}^{n}$ arranged as concentric $n$- and $m$-dimensional simplices (for $m < n$).  However, we have been unable to find a proof  that would easily generalize.

\section{Central Configurations in $\mathbb{R}^d$}
We now want to express the equations for central configurations in a way that is more convenient for our purposes. From equation (\ref{eqgamma}) it is easy to obtain
\beq 
\gamma_i-\gamma_j=(m_i+m_j)R_{ij}(q_i-q_j)+\sum_{k\neq i,j} m_k(R_{ik}(q_i-q_k)-R_{jk}(q_j-q_k)).
\label{eqgamma2}
\eeq
On the other hand the equation for the central configuration can be expressed as $\gamma_i-\gamma_j=\lambda(q_i-q_j)$.
Taking the wedge product of equation (\ref{eqgamma2}) with the vector $q_i-q_j$ we get
\beq
\sum_{k\neq i,j}m_k (R_{ik}-R_{jk})\Lambda_{ijk}=0 \label{eqcc}
\eeq
where $\Lambda_{ijk}=(q_i-q_j)\wedge(q_i-q_k)$. For a non-collinear configuration, the system of equations (\ref{eqcc}) says that $\gamma_i-\gamma_j$ and $q_i-q_j$ are linearly dependent and therefore is equivalent to the definition of central configuration. The equations (\ref{eqcc}), in the particular case of a planar central configuration, are known as the Laura-Andoyer equations and the bivector $\Lambda_{ijk}$ is simply twice the oriented area of the triangle $(q_i,q_j,q_k)$.
G. Meyer (Meyer 1933) generalized  the Laura-Andoyer equations to  to higher dimensions (see also Hagihara 1970; Albouy 2003). 

Equation (\ref{eqcc}) can be written in a slightly different way if  one  takes the wedge product of equation (\ref{eqgamma}), (with  the indexes renamed: $i=i_1$, $j=i_2$) with $(q_{i_2}-q_{i_3})\wedge (q_{i_3}-q_{i_4})\wedge \ldots\wedge(q_{i_{d-1}}-q_{i_{d}})$. Thus the equations for the central configurations in ${\mathbb R}^d$ can be written as
\beq\label{eqsimplex}
\sum_{k\neq i_1,\ldots,i_d} m_k(R_{i_1k}-R_{i_2k})\Delta_{i_1,i_2,\ldots,i_d,k}=0
\eeq
where $\Delta_{i_1,i_2,\ldots,i_d,k} = \det(q_{i_1}-q_{i_2},\ldots,q_{i_{d-1}}-q_{i_{d}},q_{i_{d}}-q_k)$ i.e. it is $d!$ times the signed volume of the $d$-dimensional simplex formed by the masses $m_{i_1},m_{i_2},$ $\ldots,m_{i_d},m_k$.

In particular, when $d=3$, the equations above can be written as 
\beq\label{eqtetra}
f_{ijh}=\sum_{k\neq i,j,h} m_k(R_{ik}-R_{jk})\Delta_{ijhk}=0
\eeq
where $\Delta_{ijhk}$ is the coefficient of $(q_i-q_j)\wedge(q_j-q_h) \cdot(q_h-q_k)$, i.e. it  is six times the signed volume of the tetrahedron formed by $m_i,m_j,m_h,m_k$. Moreover, since $f_{ijh}=-f_{jih}$, the system of equations (\ref{eqtetra}) provides $n(n-1)(n-2)/2$ equations. Usually  the $\Delta_{ijhk}$ will be replaced by the non-negative volumes $D_{ijhk}$ in order to make the sign of the terms in our equations more apparent.

\begin{remark} If $n=d+1$,  and all the mutual distances $r_{ij}$ are equal, that is, when all masses form a regular simplex of dimension $d$, then condition (\ref{eqcc}) is verified. Thus the regular simplex of dimension $d$ is a central configuration in ${\mathbb R}^d$ for any value of the masses. 
\end{remark}
\begin{figure}[thb]
  \begin{center}
\leavevmode
 \epsfxsize=5.5cm  
 \epsfbox{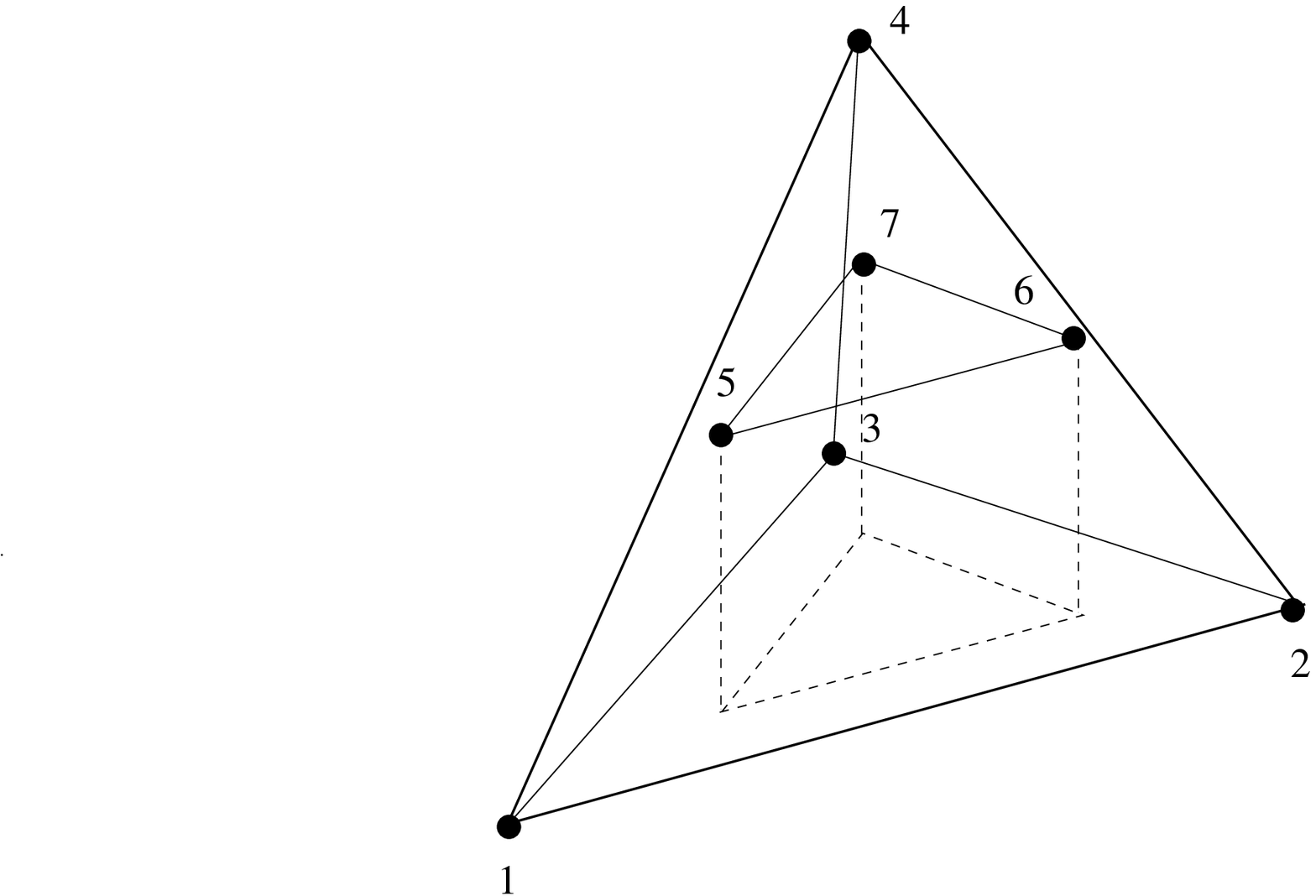}%
\caption{A typical configuration}  \label{tetrahedron}
  \end{center}
\end{figure}
\section{Stacked Central Configurations in ${\mathbb R}^3$}
We now investigate configurations with the following symmetries and relations:
\beq
\begin{split}\label{eqsymm}
r_{12}=&r_{13}=r_{23}=r_{14}=r_{24}=r_{34}=1\\
r_{15}=&r_{26}=r_{37}\\
r_{45}=&r_{46}=r_{47}\\
r_{56}=&r_{67}=r_{57}\\
r_{16}=&r_{17}=r_{25}=r_{27}=r_{35}=r_{36}.
\end{split}
\eeq

Figure \ref{tetrahedron} depicts a typical configuration.  These symmetries result in many of the $D_{ijkl}$ being equal.  These equalities will be exploited to simplify a given $D_{ijkl}$ to $D_{mnop}$ where $\{mnop\} = \{ijkl\}$ and $m < n < o < p$ ; e.g. $D_{1475}$ will be replaced by $D_{1457}$.

We want to prove the following
\begin{theorem}
For every $r_{15}\in (0,\sqrt{6}/4)$ there  are positive masses $m_1=m_2=m_3$, $m_4$ and $m_5=m_6=m_7$ and distances
which satisfy equation (\ref{eqsymm}) with  $1>r_{16}>r_{45}>r_{15}$ such that the seven points with these masses form a spatial central configuration.
\end{theorem}

In order to prove the theorem we can start observing that, since the configuration is highly symmetric,  there are only  two degrees of freedom, which can be parametrized with $r_{15}$ and $r_{16}$. Let $s_{ij}=r_{ij}^2$, then the other two distances are determined by
\[
r_{56} = r_{16}^2 - r_{15}^2 = s_{16} - s_{15}
\label{r56eq}
\]
(that indicates that the polygon of vertices $\{1,5,6,2\}$  is a isosceles trapezoid)
and $g(s_{15},s_{16},s_{45}) :=$
\beq
-3 + 2s_{15} - 3s_{15}^2 + 4s_{16} + 4s_{15}s_{16} - 4s_{16}^2 + 2s_{45} + 2s_{15}s_{45} + 4s_{16}s_{45} - 3s_{45}^2=0.
\label{rel}
\eeq
The last equation follows from imposing that the 4-volume of the pentachoron (4-simplex) of vertices $\{1,2,3,4,5\}$ is zero.
This  can be done setting the following Cayley-Menger determinant (see Sommerville 1958)  to zero:

\beq-9216V^2=\left |\begin{array}{cccccc}
0 & 1 & 1 & 1 & 1 & 1 \\ 
1 & 0 &s_{12}  &s_{13} & s_{14} & s_{15}  \\ 
1 & s_{21} & 0 & s_{23} & s_{24} & s_{25} \\ 
1 & s_{31} & s_{32} & 0 & s_{34} & s_{35} \\ 
1 & s_{41} & s_{42} & s_{43} & 0 & s_{45}\\ 
1 &s_{51} & s_{52} & s_{53} & s_{54} & 0
\end{array} \right|
\eeq
 where the $s_{ij}$ are computed for the above configuration and $V$ is the 4-volume.

Recall that, where it is convenient,  we   replace the $\Delta_{ijhk}$ by the non-negative volumes 
$D_{ijhk}$. However in order to do that we need $\Delta_{1457}\leq 0$, since this oriented volume can have different signs for the configuration defined by the equations (\ref{eqsymm}).  This corresponds to requiring that the triangle $567$ is inside the convex hull of the regular tetrahedron. 

In the following we will look for the configurations satisfying $1>r_{16}>r_{45}>r_{15}$.  Since the triangle $567$ collapses to a point at $r_{15} = \sqrt{6}/4$, we restrict our attention to  $r_{15} = \in (0,\sqrt{6}/4)$.

Let us condense the above constraints on the configuration into a region:
\begin{definition}
$\Omega \subset {\mathbb R}^3$ is the set \[\begin{split}\{ &(r_{15},r_{16},r_{45}) \ |  \ r_{15}\in (0,\sqrt{6}/4),\ \Delta_{1457}\leq 0,  \ g(r_{15}^2, r_{16}^2, r_{45}^2) = 0,\\ &1>r_{16}>r_{45}>r_{15} \}\end{split}\]
\end{definition}

Note that $\Omega$ can be written as a graph over its projection in the $(r_{15},r_{16})$ plane, and that is usually how we will think of it.  In particular, $r_{45}$ can be expressed explicitly as a function of $r_{15}$ and $r_{16}$.  In a slight abuse of notation, we may refer to seven-body configurations `in $\Omega$', by which we mean a configuration whose $r_{ij}$ satisfy equation (\ref{eqsymm}), equation (\ref{rel}), and whose $(r_{15},r_{16},r_{45})$ are in $\Omega$.

By assumption we will have $r_{15}\neq r_{16}$, $r_{15}\neq r_{17}$.
From equation (\ref{eqtetra}), with $i=1, j=2, h=3$, we obtain 
\beq
f_{123}=-m_5(R_{15}-R_{25})D_{1235}-m_6(R_{16}-R_{26})D_{1236}=0
\eeq 
and since $R_{15}=R_{26}$, $R_{16}=R_{25}$ and $D_{1235}=D_{1236}$ we have
\beq
(m_5-m_6)(R_{15}-R_{25})D_{1235} = 0 
\eeq
and we find $m_5=m_6$.
Similarly with $i=2$, $j=3$, $h=1$ we find $m_6=m_7$.
Thus $m_5=m_6=m_7$.

Also from equation (\ref{eqtetra}), with $i=5,j=6$ and $h=7$, we obtain
\beq
f_{567}=-m_1(R_{51}-R_{61})D_{1567} - m_2(R_{52}-R_{62})D_{2567}=0
\eeq
and since $D_{2567}=D_{1567}$ we have
\beq
(m_1-m_2)(R_{51}-R_{61})D_{1567}=0
\eeq
and thus $m_1=m_2$. Similarly with $i=6,j=7,h=5$ we find $m_2=m_3$.
Thus $m_1=m_2=m_3$.

Now let $i=1$, $j=4$ and $h=7$. Then from equation (\ref{eqtetra}) we obtain

\beq
f_{147}=m_5((R_{15}-R_{45})\Delta_{1457}+(R_{16}-R_{45})\Delta_{1467})=0
\eeq
As long as $m_5\neq 0$ the previous equation gives a relatively simple constraint on the geometry of our configuration:
\beq
 H:=(R_{15}-R_{45})D_{1456}-(R_{45}-R_{16})D_{1467}=0.
\label{eqh}
\eeq

Thus we are interested in the subset of $\Omega$ where $H = 0$:
\begin{definition}
Let $\Omega_H \in \Omega$ denote the subset of $\Omega$ for which $H = 0$.
\end{definition}

Before analyzing $H$ in more detail we prove two technical lemmas. 
Let $1/a$ be the distance  between the plane $P$ defined by the points 1, 2 and 3 and the plane $P'$ defined by the points 5, 6 and 7,  and let $b=r_{56}$. With this notation we can prove the following

\begin{lemma}\label{function}  Given a configuration in $\Omega$ we have  $f_{147}=b f_{241}$, $f_{175}=-bf_{251}$, $f_{645}=bf_{461}$
and $f_{275}=b(f_{251}+f_{271})$.
\end{lemma}
\begin{proof}
To prove that $f_{147}=-b f_{241}$ note that
\beq
\begin{split}
f_{241}&= m_5[(R_{16}-R_{45})(\Delta_{1245}+\Delta_{1247})
+(R_{15}-R_{45})\Delta_{1245}]\\
f_{147}&= -m_5[(R_{16}-R_{45})\Delta_{1467}+(R_{15}-R_{45})\Delta_{1457}].
\end{split}
\eeq
A simple computation shows that
\begin{align}
\Delta_{1457}&=\frac{\sqrt{3}b(3-\sqrt{6}a+\sqrt{6}ab)}{18a}=-b\Delta_{1245}\\
\Delta_{1467}&=- \frac{\sqrt{3}b(-6+\sqrt{6}ab+2\sqrt{6}a)}{18a}=-b(\Delta_{1245}+\Delta_{1247}). 
\end{align}
This concludes the first part of the proof.

To show that  $f_{175}=-bf_{251}$ consider 
\beq
\begin{split}
f_{251}=& -m_1(1-R_{16})\Delta_{1235}-m_4(1-R_{45})\Delta_{1245}\\ &+m_5(R_{16}-R_{56})\Delta_{1257}\\
f_{175}=& -m_1(1-R_{16})\Delta_{1257}+ m_4(1-R_{45})\Delta_{1456}\\ &+ m_5(R_{16}-R_{56})\Delta_{1567}; 
\end{split}
\eeq
then it is easy to find the following relationships between volumes
\beq
\begin{split}
\Delta_{1257}& =\sqrt{3}b/(2a) = - b\Delta_{1235}\\
\Delta_{1456}&= -\frac{\sqrt{3}b(3-\sqrt{6}a+\sqrt{6}ab)}{18a}= b\Delta_{1245}\\
\Delta_{1567}&= -\sqrt{3}b^2/(2a)=-b\Delta_{1257}.
\end{split}
\eeq
To show that $f_{645}=bf_{461}$ consider 
\beq
\begin{split}
f_{461}=&m_1[(1-R_{15})\Delta_{1245} - (1-R_{16})\Delta_{1247}]\\ &+m_5(R_{45}-R_{56})(-\Delta_{1456}+\Delta_{1467})\\
f_{645}=&-\{-m_1[(1-R_{15})\Delta_{1456}+(1-R_{16})(\Delta_{1456} + \Delta_{1467})]\\ &+ m_5(R_{45}-R_{56})\Delta_{4567}\}
\end{split}
\eeq
then we have the following relationships
\beq
\begin{split}
&\Delta_{1456}= -\frac{\sqrt{3}b(3-\sqrt{6}a+\sqrt{6}ab)}{18a}  = b\Delta_{1245}\\
&\Delta_{1456}+\Delta_{1467}= -\frac{\sqrt{3}b(-3+\sqrt{6}a+2\sqrt{6}ab)}{18a}  = -b\Delta_{1247}\\
&\Delta_{4567} = \frac{\sqrt{3}b^2(-3+\sqrt{6}a)}{6a}  = b(\Delta_{1456}-\Delta_{1467}).
\end{split}
\eeq
Finally to prove that $f_{275}=b(f_{251}+f_{271})$ we consider 
\[\begin{split}
f_{275}=&m_1[(1-R_{16})\Delta_{1257}+(1-R_{15})\Delta_{1257}]+m_4(1-R_{45})\Delta_{1467}\\&+m_5(R_{15}-R_{56})\Delta_{1567}
\end{split}\]

and
\beq
\begin{split} 
f_{251} =& -m_1(1-R_{16})\Delta_{1235}  - m_4(1-R_{45})\Delta_{1245}+m_5(R_{16}-R_{56})\Delta_{1257}\\
f_{271}=& -m_1(1-R_{15})\Delta_{1235} -m_4(1-R_{45})\Delta_{1247}\\ &- m_5(R_{16}+R_{15}-2R_{56})\Delta_{1257}
\end{split}
\eeq
Standard computations give the following relationships
\beq
\begin{split}
\Delta_{1257}=&\frac{\sqrt{3}b}{2a}= -b\Delta_{1235}\\
\Delta_{1467}=&\frac{\sqrt{3}b(6-2\sqrt{6}a-\sqrt{6}ab)}{18a}= -b(\Delta_{1245}+\Delta_{1247})\\
\Delta_{1567}=&-\frac{\sqrt{3}b^2}{2a}=-b\Delta_{1257}
\end{split}
\eeq
This completes the proof.
\end{proof}

Using the lemma above we have 

\begin{lemma}
Given a configuration in $\Omega_H$ there exists positive masses for which the configuration is a central configuration. 
\end{lemma}
\begin{proof}
We have seen that, to have $f_{123}=f_{231}=0$ and $f_{567}=f_{675}=0$,  we must choose $m_1=m_2=m_3$ and $m_5=m_6=m_7$. Because of the symmetry of the configuration and of the above choice of the masses we have that
\beq
\begin{split}
f_{123}&=f_{124}=f_{125}=f_{126}=f_{127}=f_{132}=f_{134}=f_{135}=f_{136}=f_{137}\\
&=f_{145}=f_{154}=f_{231}=f_{234}=f_{235}=f_{236}=f_{237}=f_{246}=f_{347}\\
&=f_{264}=f_{374}=f_{451}=f_{462}=f_{473}=f_{561}=f_{562}=f_{563}=f_{564}\\
&=f_{567}=f_{571}=f_{572}=f_{573}=f_{574}=f_{576}=f_{671}=f_{672}=f_{673}\\
&=f_{674}=f_{675}=0
\end{split}
\eeq
A first inspection of the equations shows that the  remaining functions that must vanish are $f_{147}$ and $f_{241}$, $f_{251}$ and $f_{175}$, $f_{271}$, $f_{461}$ and $f_{645}$ and $f_{275}$.
However the equations above are not independent for the highly symmetric configuration under discussion.
Inspecting carefully the remaining equations one finds that the  remaining functions that must vanish are actually only  $f_{147}$, $f_{251}$, $f_{271}$, $f_{461}$ (see Lemma \ref{function}). 

Since $H=0$, $f_{147}$ vanishes. $f_{251}$, $f_{271}$ and  $f_{461}$ form a $3\times 3$ homogeneous linear system  of the form $Ax=0$, where $x=(m_1,m_4,m_5)$ and in $\Omega_H$ the matrix $A$ takes the form

\beq\small
\left(\begin{array}{ccc}
(1-R_{16})D_{1235}&(R_{45}-1)D_{1245} &(R_{16}-R_{56})D_{1257}\\
(1-R_{15})D_{1235}&(R_{45}-1)D_{1247}&(2R_{56}-R_{16}-R_{15})D_{1257}\\
(1-R_{15})D_{1245}+(R_{16}-1)D_{1247}&0&(R_{56}-R_{45})(D_{1456}+D_{1467})\\
\end{array}
\right)
\eeq

For a nonzero mass solution $(m_1,m_4,m_5)$ to exist, the determinant of $A$ must be zero.
The fact that the determinant is zero can be proved showing that the third row is a linear combination of the first two. 
In order to do that, since $H=0$, $R_{45}$ can be written as
\[
R_{45}=\frac{R_{15}D_{1456}+R_{16}D_{1467}}{D_{1456}+D_{1467}}
\]
so
\[
(R_{56}-R_{45})(D_{1456}+D_{1467})=(R_{56}-R_{15})D_{1456}+(R_{56}-R_{16})D_{1467}.
\]
Let $v_1$, $v_2$ and $v_3$ be the rows of $A$. Then a standard computation shows that $v_3=\alpha v_1+\beta v_2$ where 
\beq
\begin{split}
\alpha &=-\frac{-3+2\sqrt{6}~ab+\sqrt{6}~a}{9}=-D_{1247}/D_{1235}\\
\beta &=-\frac{3+\sqrt{6}~ab-\sqrt{6}~a}{9}=D_{1245}/D_{1235}.
\end{split}
\eeq
This is because
\beq
\begin{split}
\alpha(1-R_{45})&D_{1245}+\beta(1-R_{45})D_{1247}\\&=-(1-R_{45})D_{1245}\frac{D_{1247}}{D_{1235}}+(1-R_{45})D_{1247}\frac{D_{1245}}{D_{1235}}=0.
\end{split}
\eeq
Moreover
\beq
\begin{split}
(R_{56}&-R_{15})D_{1257}\beta+(R_{56}-R_{16})(D_{1257}\beta-D_{1257}\alpha)\\
&=(R_{56}-R_{15})D_{1456}+(R_{56}-R_{16})D_{1467}
\end{split}
\eeq
since
\beq
D_{1257}\beta=-\frac{\sqrt{3}b(3-\sqrt{6}a+\sqrt{6}ab)}{18a}=D_{1456}
\eeq
\beq
D_{1257}\beta-D_{1257}\alpha=\frac{\sqrt{3}b(-6+2\sqrt{6}a+\sqrt{6}ab)}{18a}=D_{1467}
\eeq
Finally the fact that the masses may be chosen all positive is a consequence of the sign pattern 
of the coefficient matrix. For the configurations in $\Omega_H$ the sign pattern is
\beq
\left(\begin{array}{ccc}
-&+&-\\
-&+&*\\
-&0&+
\end{array}
\right)
\eeq

where $*$ denotes a sign that is unimportant for the discussion.
The first and last row imply that all the components of a nonzero nullvector have the same sign.
\end{proof}

\begin{figure}[t]
  \begin{center}
        \resizebox{!}{6.5cm}{\includegraphics{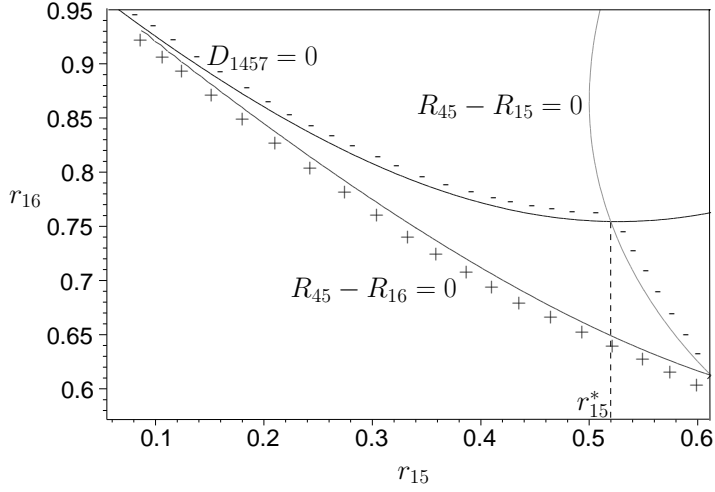}}   
    \caption{The curves $D_{1457}=0$ and $R_{45}=R_{16}$ in the plane $(r_{15},r_{16})$. In $\Omega$ $H<0$ along $D_{1456}=0$ and $H>0$ along $R_{45}=R_{16}$. The points above the curve $D_{1457}=0$ are not in $\Omega$.}
  \end{center}\label{figab}
\end{figure}
To complete  the proof of the theorem we only have to show that for every $r_{15}\in (0,\sqrt{6}/4)$ we can choose $r_{16}$ and $r_{45}$ so that $H=0$ and $(r_{15},r_{16},r_{45}) \in \Omega$. 
Note that in $\Omega$, $H$ is negative when $D_{1456}=0$ (i.e. $D_{1457}=0$, since $D_{1457}=D_{1456}$) and when $R_{45}=R_{15}$. $H$ is positive when $R_{45}=R_{16}$. 
For a given value of $0<r_{15}<r_{15}^*$  one can vary $r_{16}$ from the curve $R_{45}=R_{16}$ to the curve $D_{1457}=0$ (Figure 2
). Thus $H=0$ is satisfied for some value of $r_{16}$. For  $r_{15}^*\leq r_{15}<\sqrt{6}/4$ (i.e. when the corresponding point on the curve $R_{45}-R_{15}$ is in $\Omega$) one can vary $r_{16}$ from the curve $R_{45}=R_{16}$ to the curve $R_{45}-R_{15}=0$. Hence $H=0$ for some value of $r_{16}$.

 Therefore there is at least one central configuration for every $r_{15}\in(0,\sqrt{6}/4)$.

We can show a little more:

\begin{theorem} \label{uglytheorem}
There is only one such configuration for each $r_{15}\in (0,\sqrt{6}/4)$.
\end{theorem}
\begin{proof}
To prove that there is only one such configuration for each $r_{15}$ it is enough to show that $H'<0$, where the prime denotes differentiation with respect to $r_{16}$ and with $r_{15}$ held fixed.  
From equation (\ref{rel}), using implicit differentiation, we have
\beq
r_{45}' =-\frac{2r_{16}(1+r_{15}^2-2r_{16}^2+r_{45}^2)}{r_{45}(1+r_{15}^2+2r_{16}^2-3r_{45}^2)}
\eeq
which is negative for the configurations in $\Omega$.

Recall that $H =(R_{15}-R_{45})D_{1456}-(R_{45}-R_{16})D_{1467}$. Now consider the derivative 
\beq
\begin{split}
H'=&\ 3 r_{45}^{-4} r_{45}'D_{1456}+  (R_{15} - R_{45})D_{1456}' +( 3 r_{45}^{-4}r_{45}' - 3 r_{16}^{-4})D_{1467}\\ &- (R_{45}-R_{16})D_{1467}' \\=& \frac{1}{2 D_{1456}}\left(6 r_{45}^{-4}r_{45}' D_{1456}^2 +  (R_{15} - R_{45})(D_{1456}^2)'\right)  \\
&+ \frac{1}{2 D_{1467}}\left(6( r_{45}^{-4}r_{45}' - r_{16}^{-4})D_{1467}^2 - (R_{45}-R_{16})(D_{1467}^2)'\right)
\end{split}
\eeq
The point of the above rearrangements is to get all of the volume expressions into polynomial form.  The only parts of this expression whose signs are not clear on the region of interest are the volume derivatives  $(D_{1467}^2)'$ and $(D_{1456}^2)'$.  Unfortunately, the latter does indeed become negative within $\Omega$, and even on the set $\Omega_H = 0$.  

We will first show that $(D_{1467}^2)'$ is positive 
 on the set $\Omega_H$.  Let us write $$(D_{1467}^2)' = \frac{4 r_{16} (r_{16}^2 - r_{15}^2)t_1}{1 + r_{15}^2 + 2 r_{16}^2 - 3 r_{45}^2},$$ where $t_1$ is a polynomial in $r_{15}$, $r_{16}$, and $r_{45}$.  The sign is clearly determined by the sign of $t_1$.  If we compute the remainder $t_2$ of $t_1$ divided by the polynomial in equation (\ref{rel}), with respect to the lexicographic monomial ordering $r_{16} \succ r_{15} \succ r_{45}$, we find that $$t_2 = (1 - r_{15}^2 + r_{45}^2)\left [ 2(r_{16}^2 - r_{45}^2)(1 - r_{15}^2 + r_{45}^2) + 16 D_{145}^2 \right ],$$ where $D_{145}$ is the area of the triangle formed by the points $q_1$, $q_4$, and $q_5$.  This form of the expression is clearly positive. 

To rule out a zero of $H'$ on the set $H = 0$ we began by computing two Gr\"{o}bner bases (using Magma) of the following systems.  Both were constrained to include the set $H = 0$, by using the following two polynomials: $H$  multiplied by $(R_{15}-R_{45})D_{1456}+(R_{45}-R_{16})D_{1467}$, with denominators cleared (denoted below as $b_1$), and the condition ($\ref{rel}$) converted to the $r_{ij}$ variables (denoted as $b_2$).  Both of these polynomials also possess a number of components irrelevant to our analysis, but this is unavoidable.  To exclude components of the solutions of these systems with any of the $r_{ij} = 0$, the condition $r_{15}r_{16}r_{45}w-1 = 0$ was also added (denoted by $b_3$).  These three polynomials will be referred to as the base system:

\beq
\begin{split}
b_1 & =  -r_{45}^6 r_{15}^{10}+2 r_{16}^3 r_{45}^3 r_{15}^{10}+r_{16}^8 r_{15}^8-r_{16}^6 r_{15}^8+2 r_{16}^2 r_{45}^6 r_{15}^8-4 r_{16}^5 r_{45}^3 r_{15}^8-r_{16}^6 r_{45}^2 r_{15}^8 \\
&  -2 r_{16}^6 r_{45}^3 r_{15}^7-r_{16}^{10} r_{15}^6-r_{45}^{10} r_{15}^6+r_{16}^8 r_{15}^6+2 r_{16}^2 r_{45}^8 r_{15}^6+2 r_{45}^8 r_{15}^6+2  r_{16}^3 r_{45}^7 r_{15}^6 \\
& -2 r_{16}^4 r_{45}^6 r_{15}^6+2 r_{16}^2 r_{45}^6 r_{15}^6-r_{45}^6 r_{15}^6-4 r_{16}^5 r_{45}^5 r_{15}^6-4 r_{16}^3 r_{45}^5r_{15}^6+4 r_{16}^7 r_{45}^3 r_{15}^6 \\
& -4 r_{16}^5 r_{45}^3 r_{15}^6+2 r_{16}^3 r_{45}^3 r_{15}^6+r_{16}^8 r_{45}^2 r_{15}^6+r_{16}^6 r_{45}^2 r_{15}^6+2r_{16}^6 r_{45}^5 r_{15}^5+2 r_{16}^8 r_{45}^3 r_{15}^5 \\
& +2 r_{16}^6 r_{45}^3 r_{15}^5+r_{16}^6 r_{45}^6 r_{15}^4-2 r_{16}^6 r_{45}^7 r_{15}^3+2 r_{16}^8  r_{45}^5 r_{15}^3+2 r_{16}^6 r_{45}^5 r_{15}^3  -2 r_{16}^{10} r_{45}^3 r_{15}^3 \\
& +2 r_{16}^8 r_{45}^3 r_{15}^3-2 r_{16}^6 r_{45}^3 r_{15}^3-r_{16}^6 r_{45}^8 r_{15}^2-r_{16}^8 r_{45}^6 r_{15}^2-r_{16}^6 r_{45}^6 r_{15}^2+r_{16}^6 r_{45}^{10} \\
& - r_{16}^8 r_{45}^8-r_{16}^6 r_{45}^8+r_{16}^{10} r_{45}^6-r_{16}^8 r_{45}^6+r_{16}^6 r_{45}^6, \\
b_2 & =  3 r_{15}^4-4 r_{16}^2 r_{15}^2-2 r_{45}^2 r_{15}^2-2 r_{15}^2+4 r_{16}^4+3 r_{45}^4-4 r_{16}^2-4 r_{16}^2 r_{45}^2-2 r_{45}^2+3, \\
b_3 & =  r_{15} r_{16} r_{45} w - 1
\end{split}
\eeq

The first system also included the sign-determining factor of $6 r_{45}^{-4}r_{45}' D_{1456}^2 +  (R_{15} - R_{45})(D_{1456}^2)'$, a polynomial $p_1$ of total degree $11$ in the variables $r_{15}$, $r_{16}$, and $r_{45}$\footnote{To speed the computation, the actual initial system used in the Groebner basis computation included another polynomial in the ideal generated by $p_1$ and $b_2$, but this is theoretically redundant.}:
\beq
\begin{split}
p_1 & = 6 r_{15}^{11}-24 r_{16}^2 r_{15}^9-r_{45}^2 r_{15}^9+36 r_{16}^4 r_{15}^7+8 r_{45}^4 r_{15}^7-5 r_{16}^2   r_{45}^2 r_{15}^7 -16 r_{45}^2 r_{15}^7\\
& +r_{45}^5 r_{15}^6-30 r_{16}^6 r_{15}^5-3 r_{45}^6 r_{15}^5+18   r_{16}^4 r_{15}^5-16 r_{16}^2 r_{45}^4 r_{15}^5-2 r_{45}^4 r_{15}^5\\
& -18 r_{16}^2 r_{15}^5+22 r_{16}^4   r_{45}^2 r_{15}^5+22 r_{16}^2 r_{45}^2 r_{15}^5+r_{45}^2 r_{15}^5+6 r_{15}^5-8 r_{45}^7 r_{15}^4\\
& +5   r_{16}^2 r_{45}^5 r_{15}^4-2 r_{45}^5 r_{15}^4+12 r_{16}^8 r_{15}^3+6 r_{45}^8 r_{15}^3-18 r_{16}^6   r_{15}^3-15 r_{16}^2 r_{45}^6 r_{15}^3\\
& -8 r_{45}^6 r_{15}^3+18 r_{16}^4 r_{15}^3+26 r_{16}^4 r_{45}^4 r_{15}^3+8 r_{45}^4 r_{15}^3-6 r_{16}^2 r_{15}^3-22 r_{16}^6 r_{45}^2 r_{15}^3\\
& +4 r_{16}^4 r_{45}^2   r_{15}^3-3 r_{16}^2 r_{45}^2 r_{15}^3-2 r_{45}^2 r_{15}^3+9 r_{45}^9 r_{15}^2-2 r_{16}^2 r_{45}^7   r_{15}^2+2 r_{45}^7 r_{15}^2\\
& -4 r_{16}^4 r_{45}^5 r_{15}^2-4 r_{16}^2 r_{45}^5 r_{15}^2-r_{45}^5   r_{15}^2-6 r_{45}^{11}+9 r_{16}^2 r_{45}^9+8 r_{45}^9-8 r_{16}^4 r_{45}^7\\
& -8 r_{45}^7+4 r_{16}^6   r_{45}^5-4 r_{16}^4 r_{45}^5+3 r_{16}^2 r_{45}^5+2 r_{45}^5
\end{split}
\eeq
Since $H'$ has the form $-(a + b p_1)$ with $a,b$ positive functions, if $p_1$ is positive on $\Omega_H$ then $H' < 0$ on $\Omega_H$.  After eliminating $w$, a lexicographic term order with $r_{45} \succ r_{16} \succ r_{15}$ was used to find a polynomial $q_1(r_{15})$ of degree $404$.  The first and last few terms of $q_1$ are shown below:
\beq
\begin{split}
q_1 & = r_{15}^{404}-(804495232776247942366788446548795369/\\
& 6172585442506994000774739535727460) r_{15}^{402} \\
& +(19481733600623784042235202979633272487804463/ \\
& 2759886403053727157626401541214461915200) r_{15}^{400} + \ldots \\
& -(139/99902129761265947794751686948283396618716394291200) r_{15}^2+ \\
& (1/66601419840843965196501124632188931079144262860800)
\end{split}
\eeq

There are four zeros of $q_1$ in the interval $(0,\frac{\sqrt{6}}{4})$; a rigorous demonstration of this is possible using interval arithmetic.  These zeros are approximately at $r_1 = .5104\ldots$, $r_2 = .5384\ldots$, $r_3 = .5774\ldots$, and $r_4 = .5856\ldots$.  

For the second system, we first eliminated $D_{1467}$ from the expression of $H'$ by using the condition that $H = 0$ rearranged as $D_{1467} = D_{1456}\frac{R_{15} - R_{45}}{R_{45} - R_{16}}$.  After factoring out the positive quantity $D_{1456}$ we can write 
\begin{eqnarray*}
H'= \frac{1}{2 D_{1456}}\left[6 r_{45}^{-4}r_{45}' D_{1456}^2 +  (R_{15} - R_{45})(D_{1456}^2)' + \right.\\
 \left. \frac{R_{45} - R_{16}}{R_{15} - R_{45}}\left(6( r_{45}^{-4}r_{45}' - r_{16}^{-4})D_{1467}^2 - (R_{45}-R_{16})(D_{1467}^2)'\right) \right]
\end{eqnarray*}
Now we clear denominators from the following partial sum: $$(R_{15} - R_{45})(D_{1456}^2)' +  6 \frac{R_{45} - R_{16}}{R_{15} - R_{45}}( r_{45}^{-4}r_{45}' - r_{16}^{-4})D_{1467}^2$$
and add the numerator $p_2$ to our base system.  If this polynomial does not vanish on $\Omega_H$ then neither does $H'$.  The full polynomial $p_2$ is:

\begin{equation*}
\begin{split}
p_2 & =  24 r_{15}^6 r_{16}^{16}-60 r_{15}^8 r_{16}^{14}+4 r_{45}^8 r_{16}^{14}-36 r_{15}^6 r_{16}^{14}-8 r_{15}^3 r_{45}^5 r_{16}^{14}-32 r_{15}^6 r_{45}^2 r_{16}^{14} \\
& -24  r_{15}^6 r_{45}^3 r_{16}^{13}+60 r_{15}^{10} r_{16}^{12}-8 r_{45}^{10} r_{16}^{12}+ 60 r_{15}^8 r_{16}^{12}-4 r_{15}^2 r_{45}^8 r_{16}^{12}  -4 r_{45}^8 r_{16}^{12} \\
& +16   r_{15}^3 r_{45}^7 r_{16}^{12}+24 r_{15}^6 r_{16}^{12}+8 r_{15}^5 r_{45}^5 r_{16}^{12}+8 r_{15}^3 r_{45}^5 r_{16}^{12} +16 r_{15}^6 r_{45}^4 r_{16}^{12}\\ 
& +56   r_{15}^8 r_{45}^2 r_{16}^{12}-4 r_{15}^6 r_{45}^2 r_{16}^{12}+24 r_{15}^6 r_{45}^5 r_{16}^{11}+60 r_{15}^8 r_{45}^3 r_{16}^{11}+36 r_{15}^6 r_{45}^3   r_{16}^{11}\\
& -30 r_{15}^{12} r_{16}^{10}+9 r_{45}^{12} r_{16}^{10}-30 r_{15}^{10} r_{16}^{10}-2 r_{15}^2 r_{45}^{10} r_{16}^{10}-18 r_{15}^3 r_{45}^9 r_{16}^{10}-30   r_{15}^8 r_{16}^{10}\\
& +5 r_{15}^4 r_{45}^8 r_{16}^{10}-4 r_{15}^2 r_{45}^8 r_{16}^{10}+3 r_{45}^8 r_{16}^{10}+4 r_{15}^5 r_{45}^7 r_{16}^{10}-6 r_{15}^6   r_{16}^{10}+3 r_{15}^6 r_{45}^6 r_{16}^{10}\\
& -10 r_{15}^7 r_{45}^5 r_{16}^{10}+8 r_{15}^5 r_{45}^5 r_{16}^{10}-6 r_{15}^3 r_{45}^5 r_{16}^{10}-32 r_{15}^8 r_{45}^4   r_{16}^{10}+6 r_{15}^6 r_{45}^4 r_{16}^{10}\\
& -25 r_{15}^{10} r_{45}^2 r_{16}^{10}+8 r_{15}^8 r_{45}^2 r_{16}^{10}+9 r_{15}^6 r_{45}^2 r_{16}^{10}+6 r_{15}^6   r_{45}^7 r_{16}^9-42 r_{15}^8 r_{45}^5 r_{16}^9\\
& +6 r_{15}^6 r_{45}^5 r_{16}^9-60 r_{15}^{10} r_{45}^3 r_{16}^9-60 r_{15}^8 r_{45}^3 r_{16}^9-24 r_{15}^6   r_{45}^3 r_{16}^9+6 r_{15}^{14} r_{16}^8-6 r_{45}^{14} r_{16}^8\\
& +6 r_{15}^{12} r_{16}^8+9 r_{15}^2 r_{45}^{12} r_{16}^8+8 r_{45}^{12} r_{16}^8+12 r_{15}^3   r_{45}^{11} r_{16}^8+6 r_{15}^{10} r_{16}^8-8 r_{15}^4 r_{45}^{10} r_{16}^8\\ 
& +2 r_{15}^2 r_{45}^{10} r_{16}^8-8 r_{45}^{10} r_{16}^8-18 r_{15}^5 r_{45}^9   r_{16}^8-16 r_{15}^3 r_{45}^9 r_{16}^8+6 r_{15}^8 r_{16}^8+7 r_{15}^6 r_{45}^8 r_{16}^8\\
& -2 r_{15}^4 r_{45}^8 r_{16}^8-r_{15}^2 r_{45}^8 r_{16}^8+2 r_{45}^8   r_{16}^8+16 r_{15}^7 r_{45}^7 r_{16}^8-4 r_{15}^5 r_{45}^7 r_{16}^8+16 r_{15}^3 r_{45}^7 r_{16}^8
\end{split}
\end{equation*}
\begin{equation}
\begin{split}
& +15 r_{15}^8 r_{45}^6 r_{16}^8+8 r_{15}^6 r_{45}^6   r_{16}^8-2 r_{15}^9 r_{45}^5 r_{16}^8+4 r_{15}^7 r_{45}^5 r_{16}^8+2 r_{15}^5 r_{45}^5 r_{16}^8\\
& -4 r_{15}^3 r_{45}^5 r_{16}^8-2 r_{15}^{10} r_{45}^4   r_{16}^8-4 r_{15}^8 r_{45}^4 r_{16}^8-8 r_{15}^6 r_{45}^4 r_{16}^8+7 r_{15}^{12} r_{45}^2 r_{16}^8\\
& -14 r_{15}^{10} r_{45}^2 r_{16}^8-7 r_{15}^8 r_{45}^2   r_{16}^8+2 r_{15}^6 r_{45}^2 r_{16}^8-18 r_{15}^6 r_{45}^9 r_{16}^7-12 r_{15}^8 r_{45}^7 r_{16}^7\\
& -6 r_{15}^6 r_{45}^7 r_{16}^7+24 r_{15}^{10} r_{45}^5   r_{16}^7-6 r_{15}^8 r_{45}^5 r_{16}^7-6 r_{15}^6 r_{45}^5 r_{16}^7+30 r_{15}^{12} r_{45}^3 r_{16}^7\\
& +30 r_{15}^{10} r_{45}^3 r_{16}^7+30 r_{15}^8 r_{45}^3   r_{16}^7+6 r_{15}^6 r_{45}^3 r_{16}^7-30 r_{15}^6 r_{45}^{10} r_{16}^6-18 r_{15}^8 r_{45}^8 r_{16}^6\\
& -6 r_{15}^6 r_{45}^8 r_{16}^6+9 r_{15}^6 r_{45}^{11}   r_{16}^5+15 r_{15}^8 r_{45}^9 r_{16}^5-21 r_{15}^6 r_{45}^9 r_{16}^5+15 r_{15}^{10} r_{45}^7 r_{16}^5\\
& +12 r_{15}^8 r_{45}^7 r_{16}^5+15 r_{15}^6 r_{45}^7   r_{16}^5-9 r_{15}^{12} r_{45}^5 r_{16}^5-3 r_{15}^{10} r_{45}^5 r_{16}^5+3 r_{15}^8 r_{45}^5 r_{16}^5\\
& -3 r_{15}^6 r_{45}^5 r_{16}^5-6 r_{15}^{14} r_{45}^3   r_{16}^5-6 r_{15}^{12} r_{45}^3 r_{16}^5-6 r_{15}^{10} r_{45}^3 r_{16}^5-6 r_{15}^8 r_{45}^3 r_{16}^5\\
& +24 r_{15}^6 r_{45}^{12} r_{16}^4+42 r_{15}^8 r_{45}^{10}  r_{16}^4+6 r_{15}^{10} r_{45}^8 r_{16}^4-6 r_{15}^8 r_{45}^8 r_{16}^4-9 r_{15}^8 r_{45}^{11} r_{16}^3\\
& +3 r_{15}^{10} r_{45}^9 r_{16}^3+21 r_{15}^8 r_{45}^9  r_{16}^3-9 r_{15}^{12} r_{45}^7 r_{16}^3-6 r_{15}^{10} r_{45}^7 r_{16}^3-15 r_{15}^8 r_{45}^7 r_{16}^3\\
& +3 r_{15}^{14} r_{45}^5 r_{16}^3+3 r_{15}^{12} r_{45}^5   r_{16}^3+3 r_{15}^{10} r_{45}^5 r_{16}^3+3 r_{15}^8 r_{45}^5 r_{16}^3-9 r_{15}^6 r_{45}^{14} r_{16}^2\\
& -21 r_{15}^8 r_{45}^{12} r_{16}^2+21 r_{15}^6 r_{45}^{12}   r_{16}^2-21 r_{15}^{10} r_{45}^{10} r_{16}^2-6 r_{15}^8 r_{45}^{10} r_{16}^2-15 r_{15}^6 r_{45}^{10} r_{16}^2\\
& +3 r_{15}^{12} r_{45}^8 r_{16}^2+15 r_{15}^{10}   r_{45}^8 r_{16}^2+3 r_{15}^8 r_{45}^8 r_{16}^2+3 r_{15}^6 r_{45}^8 r_{16}^2+9 r_{15}^8 r_{45}^{14}\\
& -3 r_{15}^{10} r_{45}^{12}-21 r_{15}^8 r_{45}^{12}+9  r_{15}^{12} r_{45}^{10}+6 r_{15}^{10} r_{45}^{10}+15 r_{15}^8 r_{45}^{10}-3 r_{15}^{14} r_{45}^8\\
& -3 r_{15}^{12} r_{45}^8-3 r_{15}^{10} r_{45}^8-3 r_{15}^8 r_{45}^8
\end{split}
\end{equation}

Again we computed a Gr\"{o}bner basis, in the same manner as for the first system, and obtained a polynomial $q_2(r_{15})$ of degree $466$ which is in the ideal generated by the polynomials of the second system.   The first and last few terms of $q_2$ are:

\beq
\begin{split}
q_2 & = r_{15}^{466}- (397970325260732999941673276305737323331716006483/ \\
& 4063279255141998390265984630470563259415746780)r_{15}^{464}\\
& + (717654067077578250690275945137718952732059392991627077/\\
& 585112212740447768198301786787761109355867536320000) r_{15}^{462}+ \ldots \\
&+ (769475777092319973480299/2698353810706234317867355469145285098 \\
& 555299416887370190209448673280000)r15^2\\ 
& - (2505469532410439724481/8994512702354114392891184897150950328517 \\
& 66472295790063403149557760000)
\end{split}
\eeq

Again it is possible to rigorously compute the existence of exactly three roots in the interval $(0,\frac{\sqrt{6}}{4})$.  These are located at approximately $r_5 = .5004\ldots$, $r_6 = .5027\ldots$, and $r_7 = .5252\ldots$.  

If we consider the arrangement of the roots of the polynomials $q_1$ and $q_2$, we see that we need only verify that on the set $\Omega_H$, $p_1$ is positive for $r_{15} \in (0,r_1)$, $p_2$ is positive for $r_{15} \in (r_7, \frac{\sqrt{6}}{4})$, and that either $p_1$ or $p_2$ is positive for $r_{15} \in (r_1, r_7)$.  To do this, we can pick three rational values for $r_{15}$ each of the three intervals just described and compute the Gr\"{o}bner basis of the base system specialized to those values.  We chose to use the values $\frac{1}{2}$, $\frac{11}{21}$, and $\frac{3}{5}$.  Using interval arithmetic, it is possible to verify the appropriate signs of $p_1$ and $p_2$.  

\end{proof}

\begin{figure}[t]
  \begin{center}
    \resizebox{!}{5.5cm}{\includegraphics{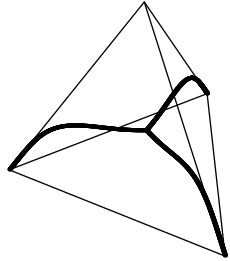}} 
    \caption{Stacked central configuration family}\label{figsevenpic}
  \end{center}
\end{figure}

\begin{figure}[t]
\begin{center}
\subfigure[]
{\label{fig:sub:a}
 \resizebox{!}{4cm}{\includegraphics{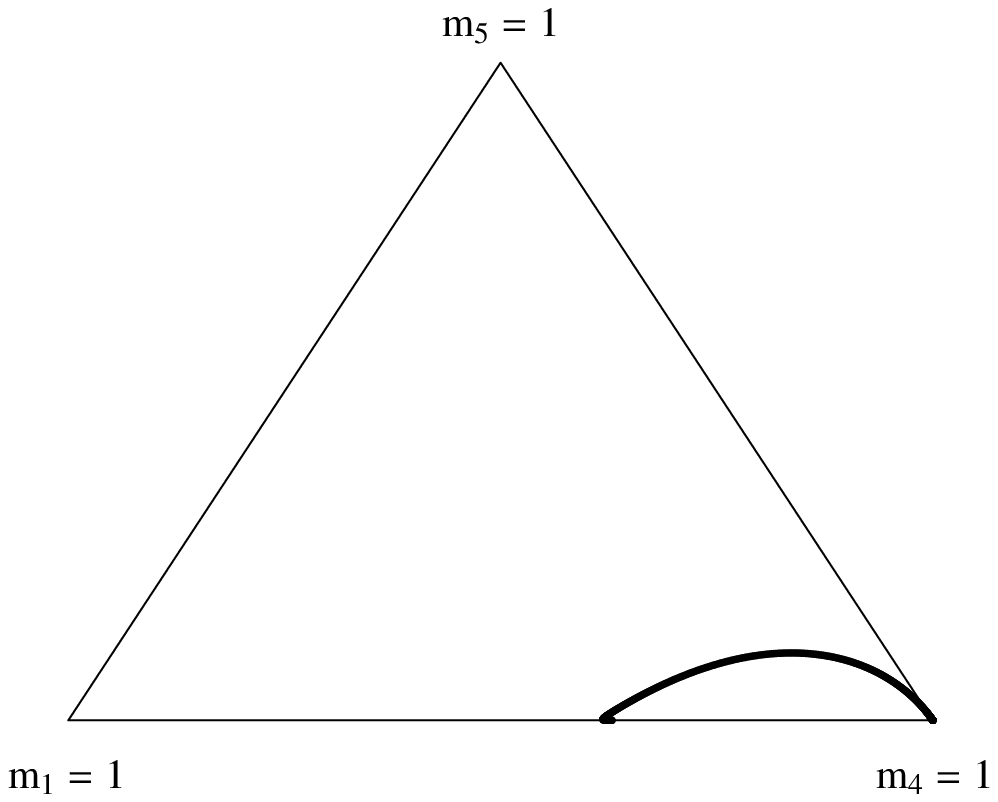}}}
\hspace{2cm}
\subfigure[]
{\label{fig:sub:b}
 \resizebox{!}{4cm}{\includegraphics{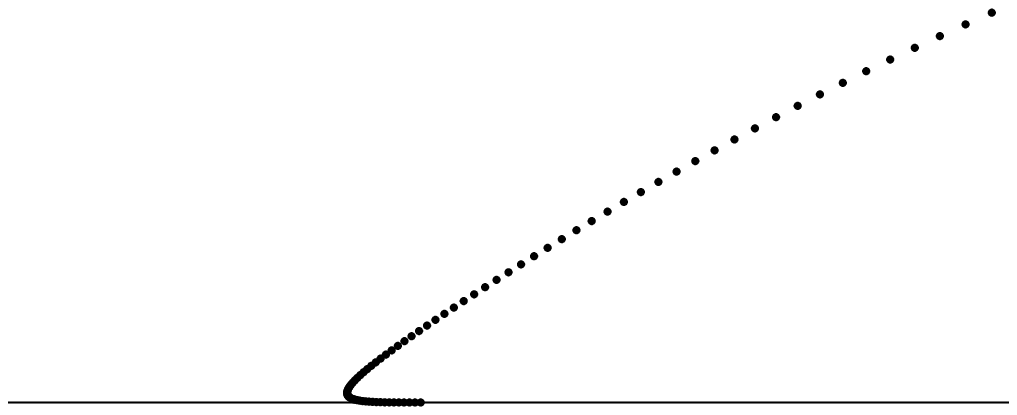}}}
\caption{(a) Masses of the configurations.  (b) Close up of masses near $m_5 = 0$.}
\label{fig:sub} 
\end{center}
\end{figure}

\section{Numerical visualizations}

Since it is often helpful to have an accurate sketch of phenomena, we reproduce here an image of the family of central configurations described above (Figure \ref{figsevenpic}), and the associated masses (Figure \ref{fig:sub}).  The masses were normalized so that $m_1 + m_4 + m_5 = 1$ and then projected into that plane. 

\section{Acknowledgements}
Marshall Hampton was partially supported by NSF grant DMS-0202268.
Manuele Santoprete was partially supported by start-up funds from 
Wilfried Laurier University.

\end{document}